\documentclass{article}  
\usepackage{url}
\usepackage{amssymb,amsmath,amsfonts,amsthm}
\usepackage{fullpage}
\newtheorem{theorem}{Theorem}

\newtheorem{definition}[theorem]{Definition}
\newtheorem{lemma}[theorem]{Lemma}

\newtheorem{corollary}[theorem]{Corollary}

\newcommand{\EquationName}[1]{\label{eq:#1}}

\newcommand{\LemmaName}[1]{\label{lem:#1}}

\newcommand{\SectionName}[1]{\label{sec:#1}}
\newcommand{\TheoremName}[1]{\label{thm:#1}}
\newcommand{\FigureName}[1]{\label{fig:#1}}

\newcommand{\Equation}[1]{Eq.\:\eqref{eq:#1}}

\newcommand{\Lemma}[1]{Lemma~\ref{lem:#1}}

\newcommand{\Section}[1]{Section~\ref{sec:#1}}
\newcommand{\Theorem}[1]{Theorem~\ref{thm:#1}}
\newcommand{\Figure}[1]{Figure~\ref{fig:#1}}

\def\FullBox{\hbox{\vrule width 8pt height 8pt depth 0pt}}

\def\qed{\ifmmode\qquad\FullBox\else{\unskip\nobreak\hfil
\penalty50\hskip1em\null\nobreak\hfil\FullBox
\parfillskip=0pt\finalhyphendemerits=0\endgraf}\fi}

\def\qedsketch{\ifmmode\Box\else{\unskip\nobreak\hfil
\penalty50\hskip1em\null\nobreak\hfil$\Box$
\parfillskip=0pt\finalhyphendemerits=0\endgraf}\fi}

\renewenvironment{proof}{\begin{trivlist} \item {\bf Proof:~~}}
  {\qed\end{trivlist}}

\newcommand{\R}{{\mathbb R}}

\newcommand{\Exp}{\mathop{\mathbb E}\displaylimits}

\newcommand{\inprod}[1]{\left\langle #1 \right\rangle}

\newcommand{\eps}{\varepsilon}

\newcommand{\oct}{\quad\quad}

\newcommand{\norm}[1]{\| #1 \|}
\newcommand{\codim}{\textrm{codim}}
\usepackage[usenames]{color}
\begin{document}
\title{On Deterministic Sketching and Streaming for Sparse Recovery
  and Norm Estimation}
\author{Jelani Nelson\thanks{Princeton
  University. \texttt{minilek@princeton.edu}. Supported by NSF
  CCF-0832797.}\oct
Huy L. Nguy$\tilde{\hat{\textnormal{e}}}$n\thanks{Princeton
  University. \texttt{hlnguyen@princeton.edu}. Supported in part by
  NSF CCF-0832797 and a Gordon Wu fellowship.}\oct
David P.\ Woodruff\thanks{IBM Almaden Research
  Center. \texttt{dpwoodru@us.ibm.com}
  }
}

\date{}

\maketitle

\renewcommand{\thefootnote}{\arabic{footnote}}
\begin{abstract}
We study classic streaming and sparse recovery problems using 
{\it deterministic} linear sketches, including 
$\ell_1/\ell_1$ and $\ell_\infty/\ell_1$
sparse recovery problems (the latter also being known as $\ell_1$-heavy
hitters), norm estimation,
and approximate inner product. 
We focus on devising a fixed matrix $A\in\R^{m\times n}$ and a
deterministic recovery/estimation procedure which work for all
possible input vectors simultaneously. Our results improve upon existing
work, the following being our main contributions:
\begin{itemize}
\item A proof that $\ell_\infty/\ell_1$ sparse recovery and inner
  product estimation are equivalent, and that incoherent matrices can
  be used to solve both problems. 
  Our upper bound for the number of measurements is
  $m=O(\eps^{-2}\min\{\log n, (\log n /
  \log(1/\eps))^2\})$. We can also obtain fast sketching and recovery
  algorithms by making use of the Fast Johnson-Lindenstrauss
  transform. Both our running
  times and number of measurements improve upon previous work. We can
  also obtain better error guarantees than previous work in terms of
  a smaller tail of the input vector.
\item A new lower bound for the number of linear measurements required
  to solve $\ell_1/\ell_1$ sparse recovery. We show $\Omega(k/\eps^2 +
  k\log(n/k)/\eps)$ measurements are required to recover an $x'$ with
  $\|x - x'\|_1 \le (1+\eps)\|x_{tail(k)}\|_1$, where $x_{tail(k)}$ is $x$
  projected onto all but its largest $k$ coordinates in magnitude.
\item A tight bound of $m=\Theta(\eps^{-2}\log(\eps^2 n))$ on the number of
  measurements required to solve deterministic norm estimation, i.e.,
  to recover $\|x\|_2 \pm \eps\|x\|_1$.
\end{itemize}

For all the problems we study, tight bounds are already known for the
randomized complexity from previous work, except in the case of
$\ell_1/\ell_1$ sparse recovery, where a nearly tight bound is
known. Our work thus aims to study the deterministic complexities of
these problems.


\end{abstract}


\section{Introduction}
In this work we provide new results for the point query problem as
well as several other related problems: approximate inner product,
$\ell_1/\ell_1$ sparse recovery, and deterministic norm estimation.
For many of these problems efficient randomized sketching and
streaming algorithms exist, and thus we are interested in
understanding the {\em deterministic} complexities of these problems.
\subsection{Applications}
Here we give a motivating application of the point query problem; for a 
formal definition of the problem, see below.
Consider $k$ servers $S^1, \ldots, S^k$, each holding a database $D^1,
\ldots, D^k$, respectively. The
servers want to compute statistics of the union $D$ of the $k$
databases. For instance, the servers
may want to know the frequency of a record or attribute-pair in
$D$. It may be too
expensive for the servers to communicate their individual databases to
a centralized server, or to compute
the frequency exactly. 
Hence, the servers wish to communicate a short summary or ``sketch''
of their databases to a
centralized server, who can then combine the sketches to answer
frequency queries about $D$. 

We model the databases as vectors $x^i \in \mathbb{R}^n$. To compute a
sketch of $x^i$, we compute 
$Ax^i$ for a matrix $A$ with $m$ rows and $n$ columns. Importantly, $m
\ll n$, and so $Ax^i$ is much easier to 
communicate than $x^i$. The servers compute $Ax^1, \ldots, Ax^k$,
respectively, and transmit these
to a centralized server. Since $A$ is a linear map, the centralized
server can compute $Ax$ for
$x = c_1 x^1 + \ldots c_k x^k$ for any real numbers $c_1, \ldots,
c_k$. Notice that the $c_i$ are allowed
to be both positive and negative, which is crucial for estimating the
frequency of record or attribute-pairs
in the difference of two datasets, which allows for tracking which
items have experienced a sudden growth
or decline in frequency. This is useful in network anomaly detection
\cite{b2,dlm02,g01,ksp03,Misra}, and also for
transactional data \cite{CM05a}. It is also useful for maintaining the
set of frequent items over a changing
database relation \cite{CM05a}. 

Associated with $A$ is an output algorithm $Out$ which given $Ax$,
outputs a vector $x'$ for which $\|x'-x\|_{\infty} \leq \eps
\|x_{tail(k)}\|_1$ for some number $k$, 
where $x_{tail(k)}$ denotes
the vector $x$ with the top $k$ entries in absolute value replaced with $0$ (the other
entries being unchanged). 
Thus $x'$ approximates $x$ well on every coordinate. We call the pair $(A,
Out)$ a solution
to the point query problem. Given such a matrix $A$ and an output
algorithm
$Out$, the centralized server can obtain an approximation to the value
of every entry in $x$, 
which depending on the application, could be the frequency of an
attribute-pair. It can also,
e.g., extract the maximum frequencies of $x$, which are useful for
obtaining the most frequent 
items. The centralized server 
obtains an entire histogram of values
of coordinates in $x$, which is a useful low-memory representation of
$x$. Notice that the communication is $mk$ words, as opposed to $nk$
if the servers were to transmit $x^1, \ldots, x^n$. 

Our correctness guarantees hold for all input vectors
simultaneously using one fixed $(A,Out)$ pair, and thus it is stronger
and should
be contrasted with the guarantee that the algorithm succeeds given
$Ax$ with high probability for some fixed input $x$. 
For example, for the point query problem,
the latter guarantee is
achieved by the CountMin sketch \cite{CM05} or CountSketch
\cite{CCF02}. One of the reasons
the randomized guarantee is less useful is because of {\it adaptive}
queries. That is, suppose the centralized server computes $x'$ and
transmits information about $x'$ to $S^1, \ldots, S^k$. Since $x'$
could depend on $A$, if the servers were to then use the same matrix
$A$ to compute sketches $Ay^1, \ldots, Ay^k$ for databases $y^1,
\ldots, y^k$ which depend on $x'$, then $A$ need not succeed, since it
is not guaranteed to be correct with high probability for inputs $y^i$
which depend on $A$.

\subsection{Notation and Problem Definitions}
Throughout this work $[n]$ denotes $\{1,\ldots,n\}$. For $q$ a prime
power, $\mathbb{F}_q$ denotes the finite field of size $q$. For
$x\in\R^n$ and $S\subseteq [n]$, $x_S$ denotes the vector with
$(x_S)_i = x_i$ for $i\in S$, and $(x_S)_i = 0$ for $i\notin S$. The
notation $x_{-i}$ is shorthand for $x_{[n]\backslash \{i\}}$. 
For a matrix $A\in \R^{m\times n}$ and integer $i\in[n]$, $A_i$
denotes the $i$th column of $A$. For matrices $A$ and vectors $x$,
$A^T$ and $x^T$ denote their transposes. For
$x\in\R^n$ and integer $k\le n$, we let $head(x, k)\subseteq [n] $
denote the set of $k$ largest coordinates in $x$ in absolute value,
and $tail(x,k) =
[n]\backslash head(x, k)$. We often use $x_{head(k)}$ to denote
$x_{head(x, k)}$, and similarly for the tail. For real numbers
$a,b,\eps\ge 0$, we use the notation $a = (1\pm\eps)b$ to convey that
$a\in [(1-\eps)b, (1+\eps)b]$. A collection of vectors
$\{C_1,\ldots,C_n\}\in [q]^t$ is called a {\em code} with {\em alphabet
  size} $q$ and {\em block length} $t$, and we define $\Delta(C_i,
C_j) = |\{k : (C_i)_k\neq (C_j)_k\}|$. The {\em relative distance} of
the code is $\max_{i\neq j} \Delta(C_i,C_j)/t$.

We now define the problems that we study in this work. In all these
problems there is some {\em error parameter} $0 < \eps < 1/2$, and we
want to design a fixed matrix $A\in\R^{m\times n}$ and deterministic
algorithm $Out$ for each problem satisfying the following.

\paragraph{Problem 1:} In the {\em $\ell_\infty/\ell_1$ recovery
  problem}, also called the {\em point query problem},
$\forall x\in\R^n$, $x' = Out(Ax)$ satisfies
$\|x - x'\|_\infty \le \eps\|x\|_1$. The pair $(A,Out)$ furthermore
satisfies the {\em $k$-tail guarantee} if actually $\|x - x'\|_\infty
\le \eps\|x_{tail(k)}\|_1$.

\paragraph{Problem 2:} In the {\em inner product
  problem},
$\forall x,y\in\R^n$, $\alpha = Out(Ax,Ay)$ satisfies
$|\alpha - \inprod{x,y}| \le \eps\|x\|_1\|y\|_1$.

\paragraph{Problem 3:} In the {\em $\ell_1/\ell_1$ recovery
  problem with the $k$-tail guarantee}, 
$\forall x\in\R^n$, $x' = Out(Ax)$ satisfies
$\|x - x'\|_1 \le (1+\eps)\|x_{tail(k)}\|_1$.

\paragraph{Problem 4:} In the {\em $\ell_2$ norm estimation problem}, 
$\forall x\in\R^n$, $\alpha = Out(Ax)$ satisfies
$|\|x\|_2 - \alpha| \le \eps\|x\|_1$.

\bigskip

We note that for the first, second, and fourth problems above, our
errors are
additive and not
relative. This is because relative error is impossible to achieve with
a sublinear number of measurements. If $A$ is a
fixed matrix with $m < n$, then it has some non-trivial kernel. Since
for all the problems above an $Out$ procedure would have to output $0$
when $Ax = 0$ to achieve bounded relative approximation, such a
procedure would fail on any input vector in the kernel which is not
the $0$ vector. 

For Problem 2 one could also ask to achieve additive error
$\eps \|x\|_p \|y\|_p$ for $p > 1$. For $y = e_i$ for a
standard unit vector $e_i$, this would
mean approximating $x_i$ up to additive error $\eps \|x\|_p$. 
This is not possible unless $m = \Omega(n^{2-2/p})$ for $1 < p \leq 2$
and $m = \Omega(n)$ for $p \geq 2$ \cite{g09b}. 

For Problem 3, it is known that the analogous
guarantee of returning $x'$ for which $\|x-x'\|_2 \leq \eps \|x_{tail(k)}\|_2$
is not possible unless $m = \Omega(n)$ \cite{cdd09}. 
%

\subsection{Our Contributions and Related Work}
We study the four problems stated above, where we have the
deterministic guarantee that a single pair $(A,Out)$ provides the
desired guarantee for all input vectors simultaneously. We first show
that point
query and inner product are equivalent up to changing $\eps$ by a
constant factor. We then show that any ``incoherent matrix''
$A$ can be used for these two problems to perform the linear
measurements; that is, a matrix
$A$ whose columns have unit $\ell_2$ norm and such that each pair of
columns has dot product at most $\eps$ in magnitude. Such matrices
can be obtained from the Johnson-Lindenstrauss (JL) lemma \cite{JL84},
almost pairwise independent sample spaces \cite{AlonGHP92,NaorN93}, or
error-correcting codes, and they play a prominent role in
compressed sensing \cite{DH01,MZ93} and mathematical approximation
theory \cite{GilbertMS03}. The connection between point
query and codes was
implicit in \cite{gm07}, though a suboptimal code was used, and
the observation that the more general class of incoherent
matrices suffices is novel.
This connection allows us to show that $m = O(\eps^{-2}\min\{\log n, (\log n /
\log(1/\eps))^2\})$ measurements suffice, and where $Out$ and the
construction of $A$ are completely deterministic.  Alon has shown
that any incoherent matrix must have $m =
\Omega(\eps^{-2}\log n / \log(1/\eps))$ \cite{a09}.
Meanwhile the best known lower bound for point query is $m =
\Omega(\eps^{-2} + \eps^{-1}\log(\eps
n))$ \cite{FPRU10,g08lb,Gluskin82}, and
the previous best known upper bound was
$m = O(\eps^{-2}\log^2 n/(\log 1/\eps + \log\log n))$ \cite{gm07}.
If the construction
of $A$ is allowed to be Las Vegas polynomial time,
then we can
use  the Fast Johnson-Lindenstrauss transforms
\cite{AC09,al09,AL11,KW11} so that $Ax$ can be computed
quickly, e.g.\ in $O(n\log m)$ time as long as $m < n^{1/2-\gamma}$
\cite{al09}, and with $m = O(\eps^{-2}\log n)$. Our $Out$ algorithm is
equally fast. We also
show that for point query, if we allow the measurement matrix $A$ to
be constructed by a polynomial Monte Carlo algorithm, then the $1/\eps^2$-tail
guarantee can be obtained essentially ``for free'', i.e.\ by 
keeping $m = O(\eps^{-2}\log n)$. Previously the work \cite{gm07} only showed how to
obtain the $1/\eps$-tail guarantee ``for free'' in this sense of not
increasing $m$ (though the $m$ in \cite{gm07} was worse). 
We note that for randomized algorithms which succeed with high probability for
any given input,
it suffices to take $m = O(\eps^{-1} \log n)$ by using the CountMin data
structure \cite{CM05}, and this is optimal \cite{jst11} (the lower bound
in \cite{jst11} is stated for the so-called heavy hitters problem, 
but also applies to the $\ell_{\infty}/\ell_1$ recovery problem).  
%

For the $\ell_1/\ell_1$ sparse recovery problem with the $k$-tail
guarantee, we show a lower bound of $m = \Omega(k\log(\eps n/k)/\eps +
k/\eps^2)$. The best upper bound is
$O(k\log(n/k)/\eps^2)$
\cite{IR08}. Our lower bound implies a separation for
the complexity of this problem in the case that one must simply pick a
random $(A,Out)$ pair which works for some given input $x$ with high
probability (i.e.\ not for all $x$ simultaneously), since \cite{PW11}
showed an $m=O(k\log n \log^3 (1/\eps)/\sqrt{\eps})$ upper bound in this case.
The first summand of our lower bound uses techniques used
in \cite{dipw10,PW11}.
The second summand uses a generalization of an argument of
Gluskin \cite{Gluskin82}, which was later rediscovered by Ganguly
\cite{g08lb}, which showed the lower bound
$m = \Omega(1/\eps^2)$ for point query.

Finally, we show how to devise an appropriate $(A,Out)$ for
$\ell_2$ norm estimation with $m = O(\eps^{-2}\log(\eps^2 n))$, which
is optimal. 
The construction of $A$ is randomized but then works for all $x$ with
high probability.
The proof
takes $A$ according to known upper bounds on Gelfand widths, and the
recovery procedure $Out$ requires solving a simple convex program. As
far as we are aware, this is the first work to investigate this
problem in the deterministic setting. In the case that $(A,Out)$ can
be chosen randomly to work for
any fixed $x$ with high probability, one can use the AMS sketch
\cite{AMS99} with $m = O(\eps^{-2}\log(1/\delta))$ to succeed with
probability $1-\delta$ and to obtain the better guarantee
$\eps\|x\|_2$. The AMS sketch can also be used for the inner product
problem to obtain error guarantee $\eps\|x\|_2\|y\|_2$ with the same
$m$.




\section{Point Query and Inner Product Estimation}\SectionName{innerprod}
We first show that the problems of point query and inner product
estimation are equivalent up to changing the error parameter $\eps$ by
a constant factor.

\begin{theorem}
Any solution $(A,Out')$ to inner product estimation with error
parameter $\eps$ yields a solution $(A, Out)$ to the point query
problem with error parameter $\eps$. Also, a solution $(A, Out)$
for point
query with error $\eps$ yields a solution $(A, Out')$
to inner product with error $12\eps$. The time complexities of $Out$
and $Out'$ are equal up to $\mathrm{poly}(n)$ factors.
\end{theorem}
\begin{proof}
Let $(A,Out')$ be a solution to the inner product problem such that
$Out'(Ax, Ay) = \inprod{x,y} \pm \eps\|x\|_1\|y\|_1$. Then given
$x\in\R^n$, to solve the point query problem we return the vector with
$Out(Ax)_i = Out'(Ax, Ae_i)$, and our guarantees are immediate.

Now let $(A,Out)$ be a solution to the point query problem. Then given
$x,y\in\R^n$, let $x' = Out(Ax), y' = Out(Ay)$. Our estimate for the
inner product is $Out'(Ax, Ay) =
\inprod{x'_{head(1/\eps)},y'_{head(1/\eps)}}$. Observe the following:
any coordinate $i$ with $|x_i'| \ge 2\eps\|x\|_1$ must have $|x_i| \ge
\eps\|x\|_1$, and thus there are at most $1/\eps$ such
coordinates. Also, any $i$ with $|x_i| \ge 3\eps\|x\|_1$ will have
$|x_i'| \ge 2\eps\|x\|_1$. Thus, $\{i : |x_i| \ge 3\eps\|x\|_1\}
\subseteq head(x', 1/\eps)$, and similarly for $x$ replaced with
$y$. Now,
\begin{align*}
\left|\inprod{x'_{head(1/\eps)},y'_{head(1/\eps)}} -
  \inprod{x,y}\right| &\le
\left|\inprod{x'_{head(1/\eps)},y'_{head(1/\eps)}} -
  \inprod{x_{head(x',1/\eps)}, y_{head(y',1/\eps)}}\right|\\
&\hspace{.2in} {}+
\left|\inprod{x_{head(x',1/\eps)}, y_{tail(y',1/\eps)}}\right| +
\left|\inprod{x_{tail(x',1/\eps)}, y_{head(y',1/\eps)}}\right|\\
&\hspace{.2in} {}+ \left|\inprod{x_{tail(x',1/\eps)},
    y_{tail(y',1/\eps)}}\right|
\end{align*}

We can bound
$$\left|\inprod{x'_{head(1/\eps)},y'_{head(1/\eps)}} -
  \inprod{x_{head(x',1/\eps)}, y_{head(y',1/\eps)}}\right|$$
by
$$\sum_{i\in head(x',1/\eps)} \eps\|y\|_1x_i + \sum_{i\in
  head(x',1/\eps)} \eps\|x\|_1y_i + \frac 1{\eps}\cdot
\eps^2\|x\|_1\|y\|_1 \le
3\eps\|x\|_1\|y\|_1 .$$

We can also bound
\begin{align*}
\left|\inprod{x_{head(x',1/\eps)}, y_{tail(y',1/\eps)}}\right| +
\left|\inprod{x_{tail(x',1/\eps)}, y_{head(y',1/\eps)}}\right| &\le
 \|x\|_1\|y_{tail(y',1/\eps)}\|_\infty +
\|x_{tail(x',1/\eps)}\|_\infty\|y\|_1\\
{}&\le 6\eps\|x\|_1\|y\|_1
\end{align*}

Finally we have the bound
\begin{equation}
 \left|\inprod{x_{tail(x',1/\eps)},
    y_{tail(y',1/\eps)}}\right| \le
\|x_{tail(x',1/\eps)}\|_2\|y_{tail(y',1/\eps)}\|_2 .\EquationName{l2}
\end{equation}
Since $\|x_{tail(x',1/\eps)}\|_\infty \le 3\eps\|x\|_1$ and
$\|x_{tail(x',1/\eps)}\|_1 \le \|x\|_1$, we have that
$\|x_{tail(x',1/\eps)}\|_2$ is maximized when it has exactly
$1/(3\eps)$ coordinates each of value exactly $3\eps\|x\|_1$, which
yields $\ell_2$ norm $\sqrt{3\eps}\|x\|_1$, and similarly for $x$
replaced with $y$. Thus the right hand side of \Equation{l2} is
bounded by $3\eps\|x\|_1\|y\|_1$. Thus in summary, our total error in
inner product estimation is $12\eps\|x\|_1\|y\|_1$.
\end{proof}

Since the two problems are equivalent up to changing $\eps$ by a
constant factor, we focus on the point query problem.
We first show that any {\it incoherent matrix}
$A$ has a correct associated
output procedure $Out$. By an incoherent matrix, we mean an $m
\times n$ matrix
$A$ for which all columns $A_i$ of $A$ have unit $\ell_2$ norm, 
and for all $i\neq j$ we have $|\inprod{A_i, A_j}| \le \eps$.
We have the following lemma. 
\begin{lemma}\LemmaName{jl}
Any incoherent matrix $A$ with error parameter $\eps$ has an
associated $\mathrm{poly}(mn)$-time deterministic
recovery procedure $Out$ for which $(A, Out)$ is a solution
to the point query problem. In fact, for any $x \in \mathbb{R}^n$,
given $Ax$ and $i \in [n]$,
the output $x'_i$ satisfies $|x'_i - x_i| \leq \eps \|x_{-i}\|_1$.
\end{lemma}
\begin{proof}
Let $x \in \mathbb{R}^n$ be arbitrary. We define $Out(Ax) = A^TAx$.
Observe that for any $i \in [n]$, we have 
$$x'_i = A_i^T A x = \sum_{j=1}^n \langle A_i, A_j \rangle x_j = x_i\pm
\eps\|x_{-i}\|_1 .$$
\end{proof}

It is known that any incoherent matrix has $m =
\Omega((\log n)/(\eps^2 \log 1/\eps))$
\cite{a09}, and the JL lemma implies such matrices
with $m = O((\log n)/\eps^2)$ \cite{JL84}. For example, there exist
matrices 
in $\{-1/\sqrt{m},1/\sqrt{m}\}^{m\times n}$ satisfying this property
\cite{a03}, which can also be found in $\mathrm{poly}(n)$ time
\cite{Sivakumar02} (we note that \cite{Sivakumar02} gives running time
exponential in precision, but the proof holds if the precision is taken
to be $O(\log (n/\eps))$.
It is also known that incoherent matrices can be obtained
from almost pairwise
independent sample spaces \cite{AlonGHP92,NaorN93} or
error-correcting codes,
and thus these tools can also be used to solve the point query
problem. The connection to codes was already implicit in \cite{gm07},
though the code used in that work is
suboptimal, as we will show soon. Below we elaborate on what
bounds these tools provide for incoherent matrices, and
what they imply for the point query problem. 

\paragraph{Incoherent matrices from JL:}
The upside of the
connection to the JL lemma is that we can obtain matrices $A$ for
the point query problem such that $Ax$ can be computed quickly, via
the Fast Johnson-Lindenstrauss Transform introduced by Ailon and
Chazelle \cite{AC09} or related subsequent works.
The JL lemma states the following.

\begin{theorem}[JL lemma]
For any $x_1,\ldots,x_N\in\R^n$ and any $0<\eps<1/2$, there exists 
$A\in\R^{m\times n}$ with $m = O(\eps^{-2}\log N)$ such
that for all $i,j\in [N]$ we have $\|Ax_i - Ax_j\|_2 =
(1\pm\eps)\|x_i - x_j\|_2$.
\end{theorem}

Consider the matrix $A$ obtained from the JL lemma when the set of
vectors is $\{0,e_1,\ldots,e_n\} \in\R^n$. Then columns $A_i$ of $A$
have $\ell_2$ norm $1\pm\eps$, and furthermore for $i\neq j$ we have
$|\inprod{A_i, A_j}| = (\|A_i - A_j\|_2^2 - \|A\|_i^2 - \|A\|_j^2)/2 =
((1\pm\eps)^22 - (1\pm\eps) - (1\pm\eps))/2 \le 2\eps + \eps^2/2$. By
scaling each column to have $\ell_2$ norm exactly $1$, we still
preserve that dot products between pairs of columns are $O(\eps)$ in
magnitude.

\paragraph{Incoherent matrices from almost pairwise independence:}
Next we elaborate on the connection between incoherent matrices and
almost pairwise independence.

\begin{definition}
An {\em $\eps$-almost $k$-wise independent sample space} is a set
$S\subseteq \{-1,1\}^n$ satisfying the following. For any $T\subseteq
[n]$, $|T| = k$, the $\ell_1$ distance between the uniform
distribution over $\{-1,1\}^k$ and the distribution of
$x(T)$ when $x$ is drawn uniformly at random from $S$ is at most
$\eps$. Here $x(T)\in\{-1,1\}^{|T|}$ is the bitstring $x$ projected
onto the coordinates in $T$.
\end{definition}

Note that if $S$ is $\eps$-almost $k$-wise independent, then for any
$|T| = k$, $|\Exp_{x\in S} \prod_{i\in T}x_i| \le \eps$. Therefore if
we choose $k=2$ and form a $|S|\times n$ matrix where the rows of $A$
are the elements of $S$, divided by a scale factor of $\sqrt{|S|}$,
then $A$ is incoherent. Known constructions of almost pairwise
independent sample spaces give $|S| = \mathrm{poly}(\eps^{-1}\log n)$
\cite{AlonGHP92,Ben-AroyaT09,NaorN93}. We do not delve into the
specific bounds on $|S|$ since they yield worse results than the
JL-based construction above. The probabilistic method
implies that such an $S$ exists with $S = O(\eps^{-2}\log n)$, matching
the JL construction, but an explicit almost pairwise independent
sample space with this size is currently not known.

\paragraph{Incoherent matrices from codes:}
Finally we explain the connection
between incoherent matrices and codes. A connection
to balanced binary codes was made in \cite{a09}, and
to arbitrary codes over larger alphabets without detail in
a remark in \cite{Alon03}. Though not novel, we elaborate on this
latter connection for the sake of completeness. Let $\mathcal{C} =
\{C_1,\ldots,C_n\}$ be a code
with alphabet size $q$, block length $t$, and relative distance
$1-\eps$. The fact that such a code gives rise to a matrix
$A\in\R^{m\times n}$ for point query with error parameter $\eps$ was
implicit in \cite{gm07}, but we make it explicit here.
We let $m = qt$ and conceptually partition the rows of $A$
arbitrarily into $t$ sets each of size $q$. For the column $A_i$, let
$(A_i)_{j, k}$ denote the entry of $A_i$ in the $k$th coordinate of
the $j$th block. We set $(A_i)_{j,k} = 1/\sqrt{t}$ if $(C_i)_j = k$, and
$(A_i)_{j,k} = 0$ otherwise. Said differently,
for $y = Ax$ we label
the entries of $y$ with double-indices $(i,j)\in [t]\times [q]$. 
We define deterministic hash functions $h_1,\ldots,h_t:[n]\rightarrow
[q]$ by $h_i(j) = (C_j)_i$, and we set $y_{i,j} = \sum_{k: h_i(k) = j}
x_k/\sqrt{t}$. Our
procedure $Out$ produces a vector $x'$ with $x'_k = \sum_{i=1}^t y_{i,
  h_i(k)}$. Each column has exactly $t$ non-zero entries of value
$1/\sqrt{t}$, and thus has $\ell_2$ norm $1$. Furthermore, for $i\neq
j$, $\inprod{A_i,A_j} = (t - \Delta(C_i, C_j))/t \le \eps$.

The work \cite{gm07} instantiated the above
with the following {\em Chinese remainder code}
\cite{KKLS94,SJJT86,WH66}. Let $p_1<\ldots<p_t$
be primes,
and let $q = p_t$. We let $(C_i)_j = i\mod p_j$. To obtain $n$
codewords with relative distance $1-\eps$, this construction
required setting $t = O(\eps^{-1}\log n / (\log(1/\eps) + \log\log
n))$ and $p_1,p_t = \Theta(\eps^{-1}\log n) = O(t\log t)$. The proof
uses that for $i,j\in[n]$, $|i-j|$ has at
most $\log_{p_1} n$ prime factors greater than or equal to $p_1$, and
thus $C_i,C_j$ can have
at most $\log_{p_1} n$  many equal coordinates. This yields $m = tq =
O(\eps^{-2}\log^2 n/(\log 1/\eps + \log\log n))$. We
observe here that this bound is never optimal. A random code with $q =
2/\eps$ and $t = O(\eps^{-1}\log n)$ has the desired properties by
applying the Chernoff bound on a pair of codewords, then a union bound
over codewords (alternatively, such a code is promised by the
Gilbert-Varshamov (GV) bound). If $\eps$ is sufficiently small, a
Reed-Solomon code performs even better. That is, we take a finite
field $\mathbb{F}_q$
for $q = \Theta(\eps^{-1}\log n/(\log\log n + \log (1/\eps)))$ and
$q=t$, and
each $C_i$ corresponds to a distinct degree-$d$ polynomial $p_i$
over $\mathbb{F}_q$ for $d = \Theta(\log n/(\log\log n + \log
(1/\eps)))$ (note there
are at least $q^d > n$ such polynomials). We set
$(C_i)_j = p_i(j)$.
The relative distance is as desired since $p_i - p_j$ has at most
$d$ roots over $\mathbb{F}_q$ and thus can be $0$ at most $d \le \eps
t$ times. This
yields $qt = O(\eps^{-2}(\log n/(\log\log n + \log (1/\eps))^2)$,
which surpasses the GV bound for $\eps < 2^{-\Omega(\sqrt{\log n})}$,
and is always better than the Chinese remainder code. We
note that this construction of a binary matrix based on Reed-Solomon
codes is identical to one used by Kautz and Singleton in the different
context of group testing \cite{KS64}.

\begin{figure}
\begin{center}
\begin{tabular}{|c|c|c|c|}
\hline
Time & $m$ & Details & Explicit?\\
\hline
$O((n\log n)/\eps^2)$ & $O(\eps^{-2}\log n)$ & $A\in
\{-1/\sqrt{m},1/\sqrt{m}\}^{m\times n}$ \cite{a03,Sivakumar02} & yes\\
\hline
$O((n\log n)/\eps)$ & $O(\eps^{-2}\log n)$ & sparse JL \cite{KN12}, GV
code & no\\
\hline
$O(nd\log^2d\log\log d/\eps)$& $O(d^2/\eps^2)$ & Reed-Solomon code & yes
\\
\hline
$O_{\gamma}(n\log m + m^{2+\gamma})$ & $O(\eps^{-2}\log n)$ &
FFT-based JL \cite{al09} & no  \\
\hline
$O(n\log n)$ & $O(\eps^{-2}\log^5 n)$ & FFT-based JL \cite{AL11,KW11}
& no\\
\hline
\end{tabular}
\end{center}
\caption{Implications for point query from JL matrices and
  codes. Time indicates the running time to compute
  $Ax$ given $x$. In the case of Reed-Solomon, $d = O(\log n/(\log\log
  n + \log(1/\eps)))$. We say the construction is ``explicit'' if $A$
  can be computed in deterministic time
  $\mathrm{poly}(n)$; otherwise we only provide a polynomial time Las
  Vegas algorithm to construct $A$.}\FigureName{jl-table}
\end{figure}

\bigskip

In \Figure{jl-table} we elaborate on what known constructions of
codes and JL matrices provide for us in terms of point query. In the
case of running time for the Reed-Solomon construction, we use that
degree-$d$ polynomials can be evaluated on $d+1$ points in a total of
$O(d\log^2d\log\log d)$ field operations over $\mathbb{F}_q$
\cite[Ch. 10]{GG99}. In the case of \cite{al09}, the constant
$\gamma>0$ can be chosen arbitrarily, and the constant in the big-Oh
depends on $1/\gamma$. We note that except in the case of Reed-Solomon
codes, the construction of $A$ is randomized (though once $A$ is
generated, incoherence can be verified in polynomial time, thus
providing a $\mathrm{poly}(n)$-time Las Vegas algorithm).


Note that \Lemma{jl} did not just give us error $\eps\|x\|_1$,
but actually gave us $|x_i - x_i'| \le \eps\|x_{-i}\|_1$, which is
stronger. We now show that an even stronger guarantee is possible. We
will show that in fact it is
possible to obtain $\|x - x'\|_\infty \le
\eps\|x_{tail(1/\eps^2)}\|_1$ while increasing $m$ by only an additive
$O(\eps^{-2}\log(\eps^2 n))$, which is less than our original
$m$ except potentially in the Reed-Solomon construction. The idea is
to, in parallel, recover a good approximation of $x_{head(1/\eps^2)}$
with error proportional to $\|x_{tail(1/\eps^2)}\|_1$ via compressed
sensing, then to subtract from $Ax$ before running our recovery
procedure. We now give details.



We in parallel run a $k$-{\it sparse recovery} algorithm which has the
following guarantee: there is a pair $(B,Out')$ such that for any $x
\in \mathbb{R}^n$, 
we have that $x' = Out'(Bx)\in\R^n$ satisfies
$\|x'-x\|_2 \leq
O(1/\sqrt{k}) \|x_{tail(k)}\|_1$.  
Such a matrix $B$ can be taken to have the {\em restricted isometry
  property of order $k$} ($k$-RIP), i.e.\ that it preserves the $\ell_2$ norm up
to a small multiplicative constant factor for all $k$-sparse vectors in
$\R^n$.\footnote{Unfortunately currently the only known constructions
  of $k$-RIP constructions with the values of $m$ we discuss are Monte
  Carlo, forcing our algorithms in this section with the $k$-tail
  guarantee to only be Monte Carlo polynomial time when constructing the
  measurement matrix.} 
It is known
\cite{gstv07} that any such $x'$ also satisfies the guarantee that
$\|x'_{head(k)}-x\|_1 \leq O(1) \|x_{tail(k)}\|_1$,
where $x'_{head(k)}$ is the vector which agrees with $x'$ on the top
$k$ coordinates in magnitude and is $0$ on the
remaining coordinates.
Moreover, it is also known \cite{bddw06} that if $B$ satisfies the JL lemma
for a particular set of $N = (en/k)^{O(k)}$ points in $\R^n$, then
$B$ will be $k$-RIP. The associated output procedure
 $Out'$ takes $Bx$ and outputs $\textrm{argmin}_{z
  \mid Bx=Bz} \|z\|_1$ by solving
a linear program \cite{CRT06}. All the JL matrices in
\Figure{jl-table} provide this
guarantee with $O(k\log(en/k))$ rows, except for the last row which
satisfies $k$-RIP
with $O(k\log(en/k)\log^2 k\log(k\log n))$ rows \cite{RV08}.
\begin{theorem}\TheoremName{main}
Let $A$ be an incoherent matrix $A$ with error parameter $\eps$,
and let $B$ be $k$-RIP. Then there is an output procedure $Out$ which
for any $x \in \mathbb{R}^n$,
given only $Ax, Bx$, outputs a vector $x'$ with $\|x'-x\|_{\infty} \leq \eps
\|x_{tail(k)}\|_1$. 
\end{theorem}
\begin{proof}
Given $Bx$, we first run the $k$-sparse recovery algorithm to obtain a vector
$y$ with $\|x-y\|_1 = O(1) \|x_{tail(k)}\|_1$. 
We then construct our output vector $x'$ coordinate by coordinate. To
construct $x'_i$, we replace $y_i$ with $0$, obtaining
the vector $z^i$. Then
we compute $A(x-z^i)$ and run the point query output procedure
associated with $A$ and index
$i$. The guarantee is that the output $w^i$ of the point query algorithm satisfies
$|w_i^i - (x-z^i)_i| \leq \eps \|(x-z^i)_{-i}\|_1$, where
$$\|(x-z^i)_{-i}\|_1 = \|(x-y)_{-i}\|_1 \leq \|x-y\|_1 = O(1)\|x_{tail(k)}\|_1,$$
and so $|(w^i+z^i)_i - x_i| = O(\eps)\|x_{tail(k)}\|_1.$
If we define our output vector by $x'_i = w^i_i+z^i_i$ and rescale $\eps$
by a constant
factor, this proves the theorem.
\end{proof}

By setting $k = 1/\eps^2$ in \Theorem{main} and stacking the rows
of a $k$-RIP and
incoherent matrix each with $O((\log n)/\eps^2)$ rows, we
obtain the following corollary, which says that by increasing the
number of measurements  $m = O(\eps^{-2}\log n)$ by only a constant
factor, we can obtain a stronger tail guarantee.
\begin{corollary}
There is an $m \times n$ matrix $A$ and associated
output procedure $Out$ which for any $x \in \mathbb{R}^n$, given $Ax$,
outputs a vector $x'$ with $\|x'-x\|_{\infty} \leq \eps
\|x_{tail(1/\eps^2)}\|_1$. Here $m = O((\log
n)/\eps^2)$.
\end{corollary}

Of course, again by using various choices of incoherent
matrices and $k$-RIP matrices, we can trade off the number of linear
measurements for various tradeoffs in the running time and tail guarantee.
It is also possible to obtain a
tail-error guarantee for inner product. While this is implied
black-box by reducing from point query with the $k$-tail guarantee, by
performing the argument from scratch we can obtain a better error
guarantee involving mixed $\ell_1$ and $\ell_2$ norms.

\begin{theorem}\TheoremName{iptail}
Suppose $1/\eps^2 < n/2$.
There is an $(A,Out)$ with $A\in\R^{m\times n}$ for $m =
O(\eps^{-2}\log n)$ such that for any $x,y\in\R^n$,
$Out(Ax,Ay)$ gives an output which is $\inprod{x,y} \pm
\eps(\|x\|_2\|y_{tail(1/\eps^2)}\|_1 + \|x_{tail(1/\eps^2)}\|_1\|y\|_2)
+ \eps^2\|x_{tail(1/\eps^2)}\|_1\|y_{tail(1/\eps^2)}\|_1$.
\end{theorem}
\begin{proof}
Using the $\ell_2/\ell_1$ sparse recovery mentioned
in \Section{innerprod}, we can recover $x',y'$ such that $\|x -
x'\|_2 \le \eps\|x_{tail(1/\eps^2)}\|_1$, and similarly for
$y-y'$. The number of measurements is the number of measurements
required for $1/\eps^2$-RIP, which is $O(\eps^{-2}\log(\eps^2 n))$.
Our estimation procedure $Out$ simply outputs $\inprod{x',y'}$. Then,
\begin{align*}
\left|\inprod{x,y} - \inprod{x',y'}\right| &= \left|\sum_i x_i(y_i -
  y_i') + y_i'(x_i - x_i')\right|\\
&{}\le \left|\sum_i x_i(y_i -
  y_i')\right| + \left|y_i'(x_i - x_i')\right|\\
&{}\le \|x\|_2\|y-y'\|_2 + \|y'\|_2\|x-x'\|_2\\
&{}\le \|x\|_2\|y-y'\|_2 + (\|y-y'\|_2 + \|y\|_2)\|x-x'\|_2
\end{align*}
The theorem then follows by our bounds on $\|x-x'\|_2$ and $\|y-y'\|_2$.
\end{proof}

Note that again $A,Out$ in \Theorem{iptail} can be taken to be applied
efficiently by using RIP matrices based on the Fast
Johnson-Lindenstrauss Transform.

\section{Lower Bound for $\ell_{\infty}/\ell_1$ Recovery}
Here we provide a lower bound for the point query problem addressed
in \Section{innerprod}.

\begin{theorem}
Let $0 < \eps < \eps_0$ for some universal constant $\eps_0 < 1$.
Suppose $1/\eps^2 < n/2$, and $A$ is an $m\times n$ matrix for which
given $Ax$ it is always
possible to produce a vector $x'$ such that $\|x-x'\|_{\infty} \le
\eps \|x_{tail(k)}\|_1$. Then $m = \Omega(k\log (n/k)/\log k +
\eps^{-2} + \eps^{-1}\log n)$.
\end{theorem}
\begin{proof}
  The lower bound of $\Omega(\eps^{-2})$ for any $k$ is already proven in~\cite{g08lb}.

The lower bound of $\Omega(k\log (n/k)/\log k + \eps^{-1}\log n)$ follows from a standard volume argument. For completeness, we give the argument below. Let $B_1(x, r)$ denote the $\ell_1$ ball centered at $x$ of radius $r$. We use the following lemma by Gilbert-Varshamov (see e.g.~\cite{dipw10}).

\begin{lemma}[{\cite[Lemma 3.1]{dipw10}}]
For any $q, k \in \mathbb{Z}^+, \eps \in \mathbb{R}^+$ with $\eps < 1 - 1/q$, there exists a set $S\subset \{0, 1\}^{qk}$ of binary vectors with exactly $k$ ones, such that $S$ has minimum Hamming distance $2\eps k$ and
$$\log |S| > (1 - H_q(\eps))k \log q$$
where $H_q$ is the $q$-ary entropy function $H_q(x) = -x \log_q \frac{x}{q-1}-(1-x)\log_q(1 - x)$.
\end{lemma}

Assume $\eps < 1/200$. Consider a set $S$ of $n$ dimensional binary vectors in $\mathbb{R}^n$ with exactly $1/(5\eps)$ ones such that minimum Hamming distance between any two vectors in $S$ is at least $1/(10\eps)$. By the above lemma, we can get $\log |S| = \Omega(\eps^{-1}\log(\eps n))$. For any $x\in S$, and $z\in B_1(x, 1/(200\eps))$, we have $\|z_{tail(k)}\|_1 \le \|z\|_1 \le 1/(5\eps)+1/(200\eps)=41/(200\eps)$, $z\in B_1(0, 41/(200\eps))$, and there are at most $4/(200\eps)$ coordinates that are ones in $x$ and smaller than $3/4$ in $z$, and at most $4/(200\eps)$ coordinates that are zeros in $x$ and at least $1/4$ in $z$. If $z'$ is a good approximation of $z$, then $\|z'-z\|_{\infty}\le 41/200 < 1/4$ so the indices of the coordinates of $z'$ at least $1/2$ differ from those of $x$ at most $8/(200\eps) < 1/(20\eps)$ places. Thus, for any two different vectors $x, y\in S$ and $z\in B_1(x, 1/(200\eps)), t\in B_1(y, 1/(200\eps))$, the outputs for inputs $z$ and $t$ are different and hence, we must have $Az \ne At$. Notice that for the mapping $x\rightarrow Ax$, the image of $B_1(x, 1/(200\eps))$ is the translated version of the image of $B_1(0, 41/(200\eps))$ scaled down in every dimension by a factor of $41$. For $x$'s in $S$, the images of $B(x,1/(200\eps))$ are disjoint subsets of the image of $B(0, 41/(200\eps))$. By comparing their volumes, we have $41^m \ge |S|$, implying $m = \Omega(\eps^{-1}\log(\eps n))$.

Next, consider the set $S'$ of all vectors in $\mathbb{R}^n$ with exactly $k$ coordinates equal to $1/k$ and the rest equal to $0$. For any $x\in S'$, and $z\in B_1(x, 1/(3k))$, we have $\|z_{tail(k)}\|_1 \le 1/(3k)$ and $z\in B_1(0, 1+1/(3k))$ centered at the origin. Therefore, if $z'$ is a good approximation of $z$, the indices of the largest $k$ coordinates of $z'$ are exactly the same as those of $x$. Thus, for any two different vectors $x, y\in S'$ and $z\in B_1(x, 1/(3k), t\in B_1(y, 1/(3k))$, the outputs for inputs $z$ and $t$ are different and hence, we must have $Az \ne At$. Notice that for the mapping $x\rightarrow Ax$, the image of $B_1(x, 1/(3k))$ is the translated version of the image of $B_1(0, 1+1/(3k))$ scaled down in every dimension by a factor of $3k+1$. For $x$'s in $S'$, the images of $B(x,1/(3k))$ are disjoint subsets of the image of $B(0, 1+1/(3k))$. By comparing their volumes, we have $(3k+1)^m \ge |S'| \ge (n/k)^k$, implying $m = \Omega(k\log(n/k)/\log k)$.

\end{proof}

\section{Lower Bounds for $\ell_1/\ell_1$ recovery}
Recall in the $\ell_1/\ell_1$-recovery problem, we would like to
design a
matrix $A\in\R^{m\times n}$ such that for any $x\in\R^n$, given $Ax$
we can recover $x'\in\R^n$ such that
$ \|x - x'\|_1 \le (1+\eps)\|x_{tail(k)}\|_1$. We now show two lower
bounds.

\begin{theorem}\TheoremName{ha}
Let $0<\eps<1/\sqrt{8}$ be arbitrary, and $k$ be an integer. Suppose
$k/\eps^2 < (n-1)/2$.
Then any matrix $A\in\R^{m\times n}$ which allows
$\ell_1/\ell_1$-recovery with the $k$-tail guarantee with error $\eps$
must have $m \ge \min\{n/2,(1/16)k/\eps^2\}$.
\end{theorem}
\begin{proof}
Without loss of generality we may assume that the rows of $A$ are
orthonormal. This is because first we can
discard rows of $A$ until the rows remaining form a basis for the
rowspace of $A$. Call this new matrix with potentially fewer rows
$A'$. Note that any dot products of rows of $A$ with
$x$ that the recovery algorithm uses can be obtained by taking linear
combinations of entries of $A'x$. Next, we can then find a matrix
$T\in \R^{m\times m}$ so that $TA'$ has orthonormal rows, and given
$TA'x$ we can recover $A'x$ in post-processing by left-multiplication
with $T^{-1}$.

We henceforth assume that the rows of $A$ are orthonormal.
Since $A\cdot 0 = 0$, and our recovery procedure must in particular be
accurate for $x=0$, the recovery procedure must output $x'= 0$ for any
$x\in\mathop{ker}(A)$. We consider $x = (I -
A^TA)y$ for $y = \sum_{i=1}^k \sigma_i e_{\pi(i)}$. Here $\pi$
is a random permutation on $n$ elements, and
$\sigma_1,\ldots,\sigma_k$ are
independent and uniform random variables in $\{-1,1\}$. Since
$x\in\mathop{ker}(A)$, which follows since $AA^T = I$ by
orthonormality of the rows of $A$, the recovery algorithm will output
$x' =
0$. Nevertheless, we will show that unless $m \ge
\min\{n/2,(1/16)k/\eps^2\}$, we
will have $\|x\|_1 > (1+\eps)\|x_{tail(k)}\|_1$ with positive
probability so that by the probabilistic method there exists
$x\in\mathop{ker}(A)$ for which $x' =
0$ is not a valid output.

If $m \ge n/2$ we are done. Otherwise, since $\|x\|_1 =
\|x_{head(k)}\|_1 +
\|x_{tail(k)}\|_1$, it is equivalent to show that $\|x_{head}(k)\|_1
> \eps\|x_{tail}(k)\|_1$ with positive probability. We first have
\allowdisplaybreaks[1]
\begin{align} 
\nonumber \Exp \|x_{tail}(k)\|_1 & \le \Exp \|x\|_1\\
\nonumber {} & \le \Exp \|y\|_1 + \Exp \|A^TAy\|_1 \\
{} & \le (\Exp \|y\|_1^2)^{1/2} + \sqrt{n}\cdot \left(\Exp
  \|A^TAy\|_2^2\right)^{1/2} \EquationName{cs}\\
\nonumber {} & = \sqrt{k} + \sqrt{n}\cdot \left(\Exp
  y^TA^TAA^TAy\right)^{1/2}\\
{} & = \sqrt{k} + \sqrt{n}\cdot \left(\Exp y^TA^TAy\right)^{1/2}
\EquationName{ortho-rows}\\
\nonumber {} & = \sqrt{k} + \sqrt{n}\cdot \left(\Exp
  \inprod{\sum_{j=1}^k\sigma_j A_{\pi(j)}, \sum_{j=1}^k\sigma_j
    A_{\pi(j)}}\right)^{1/2}\\
\nonumber {} & = \sqrt{k} + \sqrt{n}\cdot \left(\sum_{j=1}^k\Exp
  \|A_{\pi(j)}\|_2^2 \right)^{1/2}\\
\nonumber {} & = \sqrt{k} + \sqrt{kn}\cdot ( \Exp \|A_{\pi(1)}\|_2^2
)^{1/2}\\
{} & = \sqrt{k} + \sqrt{km} \EquationName{sumsquares} .
\end{align}
\Equation{cs} uses Cauchy-Schwarz. \Equation{ortho-rows} follows since
$A$ has orthonormal rows, so that $AA^T = I$. \Equation{sumsquares}
uses that the sum of squared entries over all columns equals the sum
of squared entries over rows, which is $m$ since the rows have unit
norm.

We now turn to lower bounding $\|x_{head(k)}\|_1$. Define $\eta_{i,j}
= \sigma_j/\sigma_i$ so that for fixed $i$ the $\eta_{i,j}$ are
independent and uniform $\pm 1$ random variables (except for
$\eta_{i,i}$, which is $1$). We have
\begin{align}
\nonumber \|x_{head(k)}\|_1 & \ge \|x_{\pi([k])}\|_1\\
\nonumber {} & = \sum_{i=1}^k \left|e_{\pi(i)}^Ty -
  e_{\pi(i)}^TA^Ty\right|\\
{} & = \sum_{i=1}^k \left|1 - \sum_{j=1}^k \eta_{i,j} \inprod{A_{\pi(i)},
    A_{\pi(j)}} \right|\EquationName{coord}
\end{align}

Now, for fixed $i\in[k]$ we have
\allowdisplaybreaks[1]
\begin{align}
\nonumber \Exp \left|\sum_{j=1}^k \eta_{i,j} \inprod{A_{\pi(i)},
    A_{\pi(j)}}\right| & \le \left(\Exp \left(\sum_{j=1}^k \eta_{i,j}
  \inprod{A_{\pi(i)}, A_{\pi(j)}}\right)^2\right)^{1/2}\\
\nonumber {} & = \sqrt{k}\cdot \left(\Exp \inprod{A_{\pi(1)},
    A_{\pi(2)}}^2\right)^{1/2}\\
\nonumber {} & < \sqrt{\frac{k}{n(n-1)}}\cdot \|A^TA\|_F\\
 {} & = \sqrt{\frac{k}{n(n-1)}}\cdot \|A\|_F\EquationName{projection}
\\
\nonumber {} & = \sqrt{\frac{mk}{n(n-1)}}\\
{} & < \frac 18 \EquationName{markov}
\end{align}
\Equation{projection} follows since $\|A^TA\|_F^2 =
\mathrm{trace}(A^TAA^TA) = \mathrm{trace}(A^TA) = \|A\|_F^2$. Here
$\|\cdot \|_F$ denotes the Frobenius norm, i.e.\ $\|B\|_F =
\sqrt{\sum_{i,j} B_{i,j}^2}$.

Putting things together, by \Equation{sumsquares}, a random vector $x$
has $\|x_{tail(k)}\|_1 \le 2\sqrt{k} + 2\sqrt{km} \le 4\sqrt{km}$ with
probability strictly larger than $1/2$ by Markov's inequality. Also,
call an $i\in [k]$ {\em bad} if $|x_{\pi(i)}| \le 1/2$. Combining
\Equation{coord} with \Equation{markov} and using a Markov bound we
have that the expected number of bad indices $i\in[k]$ is less than
$k/4$. Thus the probability that a random $x$ has more than $k/2$ bad
indices is at less than $1/2$ by Markov's inequality. Thus by a union
bound, with probability strictly larger than $1 - (1/2) - (1/2) = 0$,
a random $x$ taken as described simultaneously has $\|x_{tail(k)}\|_1
\le 4\sqrt{km}$ and less than $k/2$ bad indices, the latter of which
implies that
$\|x_{head(k)}\|_1 > k/2$. Thus there exists a vector in
$x\in\mathrm{ker}(A)$ for which
$\|x_{head(k)}\|_1 > \eps \|x_{tail(k)}\|_1$ when $m <
(1/16)k/\eps^2$, and we thus must have $m \ge (1/16)k/\eps^2$.
\end{proof}


We now give another lower bound via a different approach.
As in \cite{dipw10,PW11}, we use $2$-party communication complexity to prove an
$\Omega((k/\eps) \log(\eps n/k))$ bound on the number of rows of any $\ell_1/\ell_1$
sparse recovery scheme. The main difference from prior work is that we use 
deterministic communication complexity and a different communication problem.

We give a brief overview of the concepts from communication complexity that we
need, referring the reader to \cite{kn97} for further details. Formally, in the
$1$-way deterministic $2$-party communication complexity model, there are two
parties, Alice and Bob, holding inputs $x,y \in \{0,1\}^r$, respectively. The goal
is to compute a Boolean function $f(x,y)$. A single message $m(x)$ is sent from 
Alice to Bob, who then outputs $g(m(x),y)$ for a Boolean function $g$. The protocol
is correct if $g(m(x),y) = f(x,y)$ for all inputs $x$ and $y$. The $1$-way deterministic
communication complexity of $f$, denoted $D^{1-way}(f)$, is the minimum over all correct
protocols, of the maximum message length $|m(x)|$ over all inputs $x$. 

We use the $EQ(x,y) : \{0,1\}^r \times \{0,1\}^r \rightarrow \{0,1\}$ function, which is
$1$ if $x = y$ and $0$ otherwise. It is known \cite{kn97} that $D^{1-way}(EQ) = r$. We show
how to use a pair $(A, Out)$ with the property that for all vectors $z$, the output
$z'$ of $Out(Az)$ satisfies $\|z-z'\|_1 \leq (1+\eps)\|z_{tail(k)}\|_1$, to construct a correct
protocol for $EQ$ on strings $x,y \in \{0,1\}^r$ for $r = \Theta((k/\eps) \log n \log(\eps n/k))$. We then show how this implies the number of rows of $A$ is $\Omega((k/\eps) \log(\eps n/k))$.

We can assume the rows of $A$ are orthonormal as in the beginning of
the proof of \Theorem{ha}.
Let $A'$ be the matrix where 
we round each entry of $A$ to $b= O(\log n)$
bits per entry. We use the following Lemma of \cite{dipw10}.

\begin{lemma}(Lemma 5.1 of \cite{dipw10})\label{lem:round}
Consider any $m \times n$ matrix $A$ with
orthonormal rows. Let $A'$ be the result of rounding $A$ to $b$ bits per entry. Then
for any $v \in \mathbb{R}^n$ there exists an $s \in \mathbb{R}^n$ with
$A'v = A(v-s)$ and $\|s\|_1 \leq n^22^{-b}\|v\|_1$.
\end{lemma}

\begin{theorem}
Any matrix $A$ which allows $\ell_1/\ell_1$-recovery with the $k$-tail
guarantee with error $\eps$ satisfies $m = \Omega((k/\eps) \log(\eps n/k))$.
\end{theorem}
\begin{proof}
Let $S$ be the set of all strings in $\{0,c\eps/k\}^n$ containing exactly $k/(c\eps)$
entries equal to $c\eps/k$, for an absolute constant $c > 0$ specified below. 
Observe that $\log |S| = \Theta((k/\eps) \log(\eps n/k))$. 

In the $EQ(x,y)$ problem, Alice is given a string $x$ of length $r =
\log n\cdot \log|S|$. Alice splits $x$
into $\log n$ contiguous chunks $x^1, \ldots, x^{\log n}$, each containing $r/\log n$
bits. She uses $x^i$ as an index to choose an element of $S$. She sets
$$u = \sum_{i = 1}^{\log n} 2^i x^i,$$
and transmits $A'u$ to Bob.

Bob is given a string $y$ of length $r$ in the $EQ(x,y)$ problem. He performs the same
procedure as Alice, namely, he splits $y$ into $\log n$ contiguous chunks $y^1, \ldots, y^{\log n}$, 
each containing $r/\log n$ bits. He uses $y^i$ as an index to choose an element of $S$. He sets 
$$v = \sum_{i=1}^{\log n} 2^i y^i.$$
Given $A'u$, he outputs $A'(u-v)$, which by applying Lemma \ref{lem:round} once to $Au$ and once 
to $Av$, is equal to $A(u-v-s)$ 
for an $s$ with $\|s\|_1 \leq n^2 2^{-b}(\|u\|_1 + \|v\|_1) \leq 1/n$, where the last inequality
follows for sufficiently large $b = O(\log n)$.  
If $A'(u-v) = 0$, he outputs that $x$ and $y$ are equal, 
otherwise he outputs that $x$ and $y$ are not equal. 

Observe that if $x = y$, then $u = v$, and so Bob outputs the correct answer. Next, we
consider $x \neq y$, and show that $A'(u-v) \neq 0$. To do this, it
suffices to show that
$\|(u-v-s)_{head(k)}\|_1 > \eps\|u-v-s\|_1$,
as then
$Out(A(u-v-s))$ could not output
$0$, which would also mean that $A'(u-v) \neq 0$. 

To show that $\|(u-v-s)_{head(k)}\|_1 > \eps \|u-v-s\|_1$, first observe that $\|s\|_1 \leq 1/n$,
so by the triangle inequality, it is enough to show that $\|(u-v)_{head(k)}\|_1 > 2\eps\|u-v\|_1$. 

Let $z^1 = u-v$. 
Let $i \in [\log n]$ be the largest index of a chunk for which $x^{i} \neq y^{i}$, and let $j_1$ be
such that $|z^1_{j_1}| = \|z^1\|_{\infty}$. Then $|z^1_{j_1}| = c \eps \cdot 2^{i}/k$, while 
$$\|z^1\|_1 \leq 2 \cdot 2 + 2 \cdot 4 + 2 \cdot 8 + \cdots + 2 \cdot 2^i < 2 \cdot 2^{i+1} =
2^{i+2}.$$
Let $z^2$ be $z^1$ with coordinate $j_1$ removed. 
Repeating this argument on $z^2$, we again find a coordinate 
$j_2$ with $|z^{2}_{j_2}| \geq \frac{c\eps}{4k} \cdot \|z^2\|_1$. It follows by induction that after $k$
steps, and for $\eps > 0$ less than an absolute constant $\eps_0 > 0$,
$$\|(u-v)_{tail(k)}\|_1 \leq \left (1 - \frac{c\eps}{4k} \right )^k \|u-v\|_1 
\leq \left (1-c\eps \right ) \|u-v\|_1,$$
and so 
$$\|(u-v)_{head(k)}\|_1 > c \eps \|u-v\|_1.$$
Setting $c = 2$, we have that $\|(u-v)_{head(k)}\|_1 > 2\eps\|u-v\|_1$, as desired. 

Finally, observe the communication of this protocol is the number of rows of $A$ times $O(\log n)$,
since this is the number of bits required to specify $m(x) = A'u$. It follows by the communication
lower bound for $EQ$, that the number of rows of $A$ is $\Omega(r/\log n) = \Omega((k/\eps) \log(\eps n/k))$. This proves our theorem.
\end{proof}

\section{Deterministic Norm Estimation and the Gelfand Width}
\begin{theorem}
For $1 \le p < q \le \infty$, let $m$ be the minimum number such that there is an $n-m$ dimensional subspace $S$ of $\mathbb{R}^n$ satisfying $\sup_{v\in S} \frac{\norm{v}_q}{\norm{v}_p} \le \eps$. Then there is an $m\times n$ matrix $A$ and associated output procedure $Out$ which for any $x\in \mathbb{R}^n$, given $Ax$, outputs an estimate of $\norm{v}_q$ with additive error at most $\eps \norm{v}_p$. Moreover, any matrix $A$ with fewer rows will fail to perform the same task.
\end{theorem}

\begin{proof}
Consider a matrix $A$ whose kernel is such a subspace. For any sketch $z$, we need to return a number in the range $[\norm{x}_q-\eps\norm{x}_p, \norm{x}_q + \eps\norm{x}_p]$ for any $x$ satisfying $Ax=z$. Assume for contradiction that it is not possible. Then there exist $x$ and $y$ such that $Ax=Ay$ but $\norm{x}_q-\eps\norm{x}_p > \norm{y}_q + \eps\norm{y}_p$. However, since $x-y$ is in the kernel of A,
$$\norm{x}_q-\norm{y}_q \le \norm{x-y}_q \le \eps\norm{x-y}_p \le \eps(\norm{x}_p+\norm{y}_p)$$

Thus, we have a contradiction. The above argument also shows that given the sketch $z$, the output procedure can return $\min_{x:Ax=z} \norm{x}_q + \eps\norm{x}_p$. This is a convex optimization problem that can be solved using the ellipsoid algorithm. Below we give the details of the algorithm for finding a $1+\eps$ approximation of OPT.

Let $y = A^T(AA^T)^{-1}z$. Then $Ay = z = Ax$, $y$ is the projection of $x$ on the space spanned by the rows of $A$, and thus $y$ is the vector of minimum $\ell_2$ norm satisfying $Ay=z$. We have for any $x$ satisfying $Ax=z$,
\begin{equation}
n^{-1/2}\norm{y}_2\le n^{-1/2}\norm{x}_2 \le \norm{x}_q \le OPT
=\min_{x: Ax=z} \norm{x}_q+\eps\norm{x}_p \le
\norm{y}_q+\eps\norm{y}_p\le (1+\eps)\sqrt{n}\norm{y}_2\EquationName{ellipsoid}
\end{equation}

The value $\norm{y}_2$ can be computed from the sketch $z$, and we use
this value
to find OPT using binary search. Specifically, in each step we use
the ellipsoid algorithm to solve the feasibility problem
$\norm{x}_q+\eps\norm{x}_p \le M$ on the affine subspace $Ax = z$. 
Recall that when solving feasibility problems, the ellipsoid algorithm
takes time polynomial in the dimension,
the running time of a separation oracle, and the logarithm of the
ratio of volumes of an initial ellipsoid containing a feasible point
and the volume of the intersection of that ellipsoid with the feasible set.
Let
$x^*$ be the optimal solution of the minimization problem. If $M \ge
(1+\eps)OPT$ then by the triangle inequality every point in the $\ell_2$
ball centered at $x^*$ of radius $\frac{\eps n^{-1}\|y\|_2}{1+\eps}$
is feasible. Furthermore, by \Equation{ellipsoid} the set of feasible
solutions is contained in the intersection of the $\ell_2$ ball about the
origin of radius $(1+\eps)n\|y\|_2$ and the affine subspace (or
equivalently, the $\ell_2$ ball about $y$ of radius $\sqrt{(1+\eps)^2
  n^2-1}\|y\|_2$ and the affine subspace). Thus, the ellipsoid algorithm
runs in time polynomial in $n$ and $\log(1/\eps)$ assuming a
polynomial time separation oracle.

Now we describe the separation oracle. Consider a point $x$ such that
$\norm{x}_q+\eps\norm{x}_p > M$. We want to find a hyperplane
separating $x$ and $\{y| \|y\|_q + \eps\|y\|_p \le M\}$. Without loss
of generality assume
that $x_i \ge 0$ for all $i$. Define $f_{x, p, i}$ as follows:
$$f_{x, p, i} = \begin{cases}
\|x\|_p^{1-p}x_i^{p-1} &\mbox{if } p < \infty\\
1/k &\mbox{if } p = \infty \mbox{ and } x_i = \max_j x_j \mbox{ and } k = |\{t|x_t = \max_j x_j\}|\\
0 &\mbox{if } p = \infty \mbox{ and } x_i < \max_j x_j
\end{cases} .$$

The hyperplane we consider is $h\cdot y = h\cdot x$ where $h_i = f_{x,q,i} +\eps f_{x, p, i}$.

\begin{lemma}
  If $h\cdot y \ge h\cdot x$ then $\|y\|_q + \eps\|y\|_p \ge \|x\|_q + \eps\|y\|_p$.
\end{lemma}
\begin{proof}
  For any $y$, consider $y'$ such that $y'_i = |y_i|$. We have $\|y'\|_q+\eps\|y'\|_p = \|y\|_q + \eps\|y\|_p$ and $h\cdot y' \ge h\cdot y$. Thus, we only need to prove the claim for $y$ such that $y_i \ge 0~\forall i$.

  If $p < \infty$ then by H\"{o}lder's inequality,

$$\|y\|_p\cdot \|x\|_p^{p-1} = \|y\|_p\cdot
\|(x_i^{p-1})_i\|_{p/(p-1)} \ge \sum_i y_i x_i^{p-1} .$$

  If $p = \infty$ then $\|y\|_{\infty} \ge \sum_{i: x_i = \max_j x_j} y_i/k$.

  In either case, $\|y\|_p \ge \sum_i y_i f_{x, p, i}$, and the same
  inequality holds for $p$ replaced with $q$. Thus,
$$\|y\|_q + \eps\|y\|_p \ge y \cdot h \ge x\cdot h = \|x\|_q +
\eps\|x\|_p .$$
\end{proof}

By the above lemma, $h$ separates $x$ and the set of feasible solutions. This concludes the description of the algorithm.

For the lower bound, consider a matrix $A$ with fewer than $m$ rows. Then in the kernel of $A$, there exists $v$ such that $\norm{v}_q > \eps\norm{v}_p$. Both $v$ and the zero vector give the same sketch (a zero vector). However, by the stated requirement, we need to output $0$ for the zero vector but some positive number for $v$. Thus, no matrix $A$ with fewer than $m$ rows can solve the problem.
\end{proof}

The subspace $S$ of highest dimension of $\mathbb{R}^n$ satisfying $\sup_{v\in S} \frac{\norm{v}_q}{\norm{v}_p} \le \eps$ is related to the Gelfand width, a well-studied notion in functional analysis.
\begin{definition}
Fix $p < q$. The  Gelfand width of order $m$ of $\ell_p$ and $\ell_q$ unit balls in $\mathbb{R}^n$ is defined as 
$$\inf_{\textnormal{subspace }A:\codim(A)=m} \sup_{v\in A} \frac{\norm{v}_q}{\norm{v}_p}$$
\end{definition}

Using known bounds for the Gelfand width for $p=1$ and $q=2$, we get the following corollary.

\begin{corollary}
Assume that $1/\eps^2 < n/2$. There is an $m \times n$ matrix $A$ and associated output procedure $Out$ which for any $x \in \mathbb{R}^n$, given $Ax$, outputs an estimate $e$ such that $\norm{x}_2-\eps\norm{x}_1 \le e \le \norm{x}_2+\eps\norm{x}_1$. Here $m = O(\eps^{-2}\log (\eps^2 n))$ and this bound for $m$ is tight.
\end{corollary}
\begin{proof}
The corollary follows from the following bound on the Gelfand width by 
Foucart et al.~\cite{FPRU10} and Garnaev and Gluskin~\cite{GG84}:
$$\inf_{\textnormal{subspace }A:\codim(A)=m} \sup_{v\in A} \frac{\norm{v}_2}{\norm{v}_1} = \Theta\left(\sqrt{\frac{1+\log(n/m)}{m}}\right)$$
\end{proof}

\section*{Acknowledgments}
We thank Raghu Meka for answering several questions about almost
$k$-wise independent sample spaces. We thank an anonymous reviewer for
pointing out the connection between incoherent matrices and
$\eps$-biased spaces, which are used to construct almost $k$-wise
independent sample spaces.

\bibliographystyle{plain}
\bibliography{stoc}

\begin{thebibliography}{10}

\bibitem{a03}
Dimitris Achlioptas.
\newblock Database-friendly random projections: {Johnson-Lindenstrauss} with
  binary coins.
\newblock {\em J. Comput. Syst. Sci.}, 66(4):671--687, 2003.

\bibitem{AC09}
Nir Ailon and Bernard Chazelle.
\newblock The fast {Johnson-Lindenstrauss} transform and approximate nearest
  neighbors.
\newblock {\em SIAM J. Comput.}, 39(1):302--322, 2009.

\bibitem{al09}
Nir Ailon and Edo Liberty.
\newblock Fast dimension reduction using {Rademacher} series on dual {BCH}
  codes.
\newblock {\em Discrete {\&} Computational Geometry}, 42(4):615--630, 2009.

\bibitem{AL11}
Nir Ailon and Edo Liberty.
\newblock Almost optimal unrestricted fast {Johnson-Lindenstrauss} transform.
\newblock In {\em Proceedings of the 22nd Annual ACM-SIAM Symposium on Discrete
  Algorithms (SODA)}, pages 185--191, 2011.

\bibitem{Alon03}
Noga Alon.
\newblock Problems and results in extremal combinatorics - {I}.
\newblock {\em Discrete Mathematics}, 273(1-3):31--53, 2003.

\bibitem{a09}
Noga Alon.
\newblock Perturbed identity matrices have high rank: Proof and applications.
\newblock {\em Combinatorics, Probability {\&} Computing}, 18(1-2):3--15, 2009.

\bibitem{AlonGHP92}
Noga Alon, Oded Goldreich, Johan H{\aa}stad, and Ren{\'e} Peralta.
\newblock Simple construction of almost k-wise independent random variables.
\newblock {\em Random Struct. Algorithms}, 3(3):289--304, 1992.

\bibitem{AMS99}
Noga Alon, Yossi Matias, and Mario Szegedy.
\newblock {The Space Complexity of Approximating the Frequency Moments}.
\newblock {\em JCSS}, 58(1):137--147, 1999.

\bibitem{dipw10}
Khanh~Do Ba, Piotr Indyk, Eric Price, and David~P. Woodruff.
\newblock Lower bounds for sparse recovery.
\newblock In {\em SODA}, pages 1190--1197, 2010.

\bibitem{bddw06}
Richard Baraniuk, Mark~A. Davenport, Ronald DeVore, and Michael Wakin.
\newblock A simple proof of the {Restricted} {Isometry} {Property}.
\newblock {\em Constructive Approximation}, 28(3):253--263, 2008.

\bibitem{b2}
Daniel Barbar\'{a}, Ningning Wu, and Sushil Jajodia.
\newblock Detecting novel network intrusions using {Bayes} estimators.
\newblock In {\em Proceedings of the 1st SIAM International Conference on Data
  Mining}, 2001.

\bibitem{Ben-AroyaT09}
Avraham Ben-Aroya and Amnon Ta-Shma.
\newblock Constructing small-bias sets from algebraic-geometric codes.
\newblock In {\em FOCS}, pages 191--197, 2009.

\bibitem{CRT06}
Emmanuel Cand\`{e}s, Justin Romberg, and Terence Tao.
\newblock Robust uncertainty principles: Exact signal reconstruction from
  highly incomplete frequency information.
\newblock {\em IEEE Trans. Information Theory}, 52(2):489--509, 2006.

\bibitem{CCF02}
Moses Charikar, Kevin Chen, and Martin Farach-Colton.
\newblock Finding frequent items in data streams.
\newblock {\em Theor. Comput. Sci.}, 312(1):3--15, 2004.

\bibitem{cdd09}
Albert Cohen, Wolfgang Dahmen, and Ronald~A. DeVore.
\newblock Compressed sensing and best k-term approximation.
\newblock {\em J. Amer. Math. Soc.}, 22:211--231, 2009.

\bibitem{CM05}
Graham Cormode and S.~Muthukrishnan.
\newblock An improved data stream summary: the count-min sketch and its
  applications.
\newblock {\em J. Algorithms}, 55(1):58--75, 2005.

\bibitem{CM05a}
Graham Cormode and S.~Muthukrishnan.
\newblock What's hot and what's not: tracking most frequent items dynamically.
\newblock {\em ACM Trans. Database Syst.}, 30(1):249--278, 2005.

\bibitem{dlm02}
Erik~D. Demaine, Alejandro L{\'o}pez-Ortiz, and J.~Ian Munro.
\newblock Frequency estimation of {Internet} packet streams with limited space.
\newblock In {\em ESA}, pages 348--360, 2002.

\bibitem{DH01}
David~L. Donoho and Xiaoming Huo.
\newblock Uncertainty principles and ideal atomic decomposition.
\newblock {\em IEEE Trans. Inform. Th.}, 47:2558--2567, 2001.

\bibitem{FPRU10}
Simon Foucart, Alain Pajor, Holger Rauhut, and Tino Ullrich.
\newblock The {Gelfand} widths of $\ell_p$-balls for $0 < p \le 1$.
\newblock {\em Journal of Complexity}, 26(6):629--640, 2010.

\bibitem{g08lb}
Sumit Ganguly.
\newblock Lower bounds on frequency estimation of data streams.
\newblock In {\em CSR}, pages 204--215, 2008.
\newblock Full version at
  http://www.cse.iitk.ac.in/users/sganguly/csr-full.pdf.

\bibitem{g09b}
Sumit Ganguly.
\newblock Deterministically estimating data stream frequencies.
\newblock In {\em COCOA}, pages 301--312, 2009.

\bibitem{gm07}
Sumit Ganguly and Anirban Majumder.
\newblock {CR}-precis: A deterministic summary structure for update data
  streams.
\newblock In {\em ESCAPE}, pages 48--59, 2007.

\bibitem{GG84}
Andrej~Y. Garnaev and Efim~D. Gluskin.
\newblock On the widths of the {Euclidean} ball.
\newblock {\em Soviet Mathematics Doklady}, 30:200--203, 1984.

\bibitem{g01}
Anna~C. Gilbert, Yannis Kotidis, S.~Muthukrishnan, and Martin~J. Strauss.
\newblock Quicksand: Quick summary and analysis of network data.
\newblock DIMACS Technical Report 2001-43, 2001.

\bibitem{GilbertMS03}
Anna~C. Gilbert, S.~Muthukrishnan, and Martin Strauss.
\newblock Approximation of functions over redundant dictionaries using
  coherence.
\newblock In {\em SODA}, pages 243--252, 2003.

\bibitem{gstv07}
Anna~C. Gilbert, Martin~J. Strauss, Joel~A. Tropp, and Roman Vershynin.
\newblock One sketch for all: fast algorithms for compressed sensing.
\newblock In {\em STOC}, pages 237--246, 2007.

\bibitem{Gluskin82}
Efim~D. Gluskin.
\newblock On some finite-dimensional problems in the theory of widths.
\newblock {\em Vestn. Leningr. Univ. Math.}, 14:163--170, 1982.

\bibitem{IR08}
Piotr Indyk and Milan Ru\v{z}i\'{c}.
\newblock Near-optimal sparse recovery in the {$L_1$} norm.
\newblock In {\em FOCS}, pages 199--207, 2008.

\bibitem{JL84}
William~B. Johnson and Joram Lindenstrauss.
\newblock Extensions of {Lipschitz} mappings into a {Hilbert} space.
\newblock {\em Contemporary Mathematics}, 26:189--206, 1984.

\bibitem{jst11}
Hossein Jowhari, Mert Saglam, and G{\'a}bor Tardos.
\newblock Tight bounds for {$L_p$} samplers, finding duplicates in streams, and
  related problems.
\newblock In {\em PODS}, pages 49--58, 2011.

\bibitem{KN12}
Daniel~M. Kane and Jelani Nelson.
\newblock Sparser {Johnson-Lindenstrauss} transforms.
\newblock In {\em SODA}, pages 1195--1206, 2012.

\bibitem{ksp03}
Richard~M. Karp, Scott Shenker, and Christos~H. Papadimitriou.
\newblock A simple algorithm for finding frequent elements in streams and bags.
\newblock {\em ACM Trans. Database Syst.}, 28:51--55, 2003.

\bibitem{KS64}
William~H. Kautz and Richard~C. Singleton.
\newblock Nonrandom binary superimposed codes.
\newblock {\em IEEE Trans. Inf. Theory}, 10:363--377, 1964.

\bibitem{KW11}
Felix Krahmer and Rachel Ward.
\newblock New and improved {J}ohnson-{L}indenstrauss embeddings via the
  {R}estricted {I}sometry {P}roperty.
\newblock {\em SIAM J. Math. Anal.}, 43(3):1269--1281, 2011.

\bibitem{KKLS94}
Hari Krishna, Bal Krishna, Kuo-Yu Lin, and Jenn-Dong Sun.
\newblock {\em Computational Number Theory and Digital Signal Processing: Fast
  Algorithms and Error Control Techniques}.
\newblock CRC, Boca Raton, FL, 1994.

\bibitem{kn97}
Eyal Kushilevitz and Noam Nisan.
\newblock {\em Communication complexity}.
\newblock Cambridge University Press, 1997.

\bibitem{MZ93}
St\'{e}phane~G. Mallat and Zhifeng Zhang.
\newblock Matching pursuits with time-frequency dictionaries.
\newblock {\em IEEE Trans. Signal Process.}, 41(12):3397--3415, 1993.

\bibitem{Misra}
Jayadev Misra and David Gries.
\newblock Finding repeated elements.
\newblock {\em Sci. Comput. Program.}, 2(2):143--152, 1982.

\bibitem{NaorN93}
Joseph Naor and Moni Naor.
\newblock Small-bias probability spaces: Efficient constructions and
  applications.
\newblock {\em SIAM J. Comput.}, 22(4):838--856, 1993.

\bibitem{PW11}
Eric Price and David~P. Woodruff.
\newblock (1 + eps)-approximate sparse recovery.
\newblock In {\em FOCS}, pages 295--304, 2011.

\bibitem{RV08}
Mark Rudelson and Roman Vershynin.
\newblock On sparse reconstruction from {Fourier} and {Gaussian} measurements.
\newblock {\em Communications on Pure and Applied Mathematics}, 61:1025--1045,
  2008.

\bibitem{Sivakumar02}
D.~Sivakumar.
\newblock Algorithmic derandomization via complexity theory.
\newblock In {\em STOC}, pages 619--626, 2002.

\bibitem{SJJT86}
Michael~A. Soderstrand, W.~Kenneth Jenkins, Graham~A. Jullien, and Fred~J.
  Taylor.
\newblock {\em Residue Number System Arithmetic: Modern Applications in Digital
  Signal Processing}.
\newblock IEEE Press, New York, 1986.

\bibitem{GG99}
Joachim von~zur Gathen and J\"{u}rgen Gerhard.
\newblock {\em Modern Computer Algebra}.
\newblock Cambridge University Press, 1999.

\bibitem{WH66}
Richard~W. Watson and Charles~W. Hastings.
\newblock Self-checked computation using residue arithmetic.
\newblock {\em Proc. IEEE}, 4(12):1920--1931, 1966.

\end{thebibliography}



\end{document}